\ifdefined\directlua\let\pdfoutput\outputmode\fi 
\documentclass[a4paper,twocolumn,11pt,unpublished]{quantumarticle}
\pdfoutput=1
\usepackage{iftex}
\ifPDFTeX
  \usepackage[T1]{fontenc}
  \usepackage[utf8]{inputenc}
\else
  \usepackage{fontspec}
\fi
\usepackage[english]{babel}
\usepackage{amsmath,amssymb,amsthm,mathtools,bm}
\usepackage{graphicx}
\usepackage{booktabs}
\usepackage{array}
\usepackage{multirow}
\usepackage{xcolor}
\usepackage{hyperref}
\usepackage{enumitem}
\usepackage{microtype}
\usepackage{needspace}
\usepackage{textcomp}
\hypersetup{hidelinks,pdftitle={Beyond Orbital Rotations: Correlation-Rank Limits and Clifford-Accessible Measurement, from Algebra and Global Optimization},pdfauthor={Federico Zahariev and Vassiliki-Alexandra Glezakou}}

\newtheorem{theorem}{Theorem}
\newtheorem{proposition}{Proposition}
\newtheorem{corollary}{Corollary}
\newtheorem{lemma}{Lemma}
\newtheorem{definition}{Definition}
\newtheorem{remark}{Remark}
\newtheorem{observation}{Observation}

\newcommand{\F}{\mathbb{F}}
\newcommand{\CC}{\mathrm{CC}}
\newcommand{\CSA}{\mathrm{CSA}}
\DeclareMathOperator{\Var}{Var}
\DeclareMathOperator{\rank}{rank}

\newcommand{\Tr}{\operatorname{Tr}}

\begin{document}
\title{Beyond Orbital Rotations: Correlation-Rank Limits and Clifford-Accessible Measurement, from Algebra and Global Optimization}
\author{Federico Zahariev}
\email[Corresponding author: ]{fzahari@iastate.edu}
\orcid{0000-0002-3223-3576}
\affiliation{Chemical Sciences Division, Oak Ridge National Laboratory, Oak Ridge, Tennessee 37831, USA}
\affiliation{Department of Chemistry and Ames National Laboratory, Iowa State University, Ames, Iowa 50011, USA}
\author{Vassiliki-Alexandra Glezakou}
\email[Corresponding author: ]{glezakouva@ornl.gov}
\affiliation{Chemical Sciences Division, Oak Ridge National Laboratory, Oak Ridge, Tennessee 37831, USA}
\date{October 7, 2026}
\maketitle
\begin{abstract}One should not assume that reconstructing an electronic observable with a smaller set of measurement settings will cut the number of repetitions in proportion. Each setting is simply a basis change and qubit readout, but the savings are not guaranteed to be commensurate. Take two electrons of opposite spin in two spatial orbitals: the orbital degree of freedom of each electron acts as a logical qubit and rotations are independent in the spin sectors. From the observable's one- and two-qubit coefficients we derive an exact floor for the setting count and a decomposition to meet it in no more than three settings. Should the number of settings be limited, the minimum error in the two-qubit correlation coefficients is set by the singular values left out. A case in point is an exchange with equal nonzero couplings on the three Pauli axes. This demands three orbital settings yet only one joint measurement. We offer a Clifford circuit for four physical qubits to effect the necessary mapping of Pauli products under conjugation. In the case of a one-parameter family of logical singlet--triplet mixtures, the orbital-to-Clifford shot cost ratio is three at maximum mixing and trends toward 3/2 near the logical singlet. Where spin sectors may be mixed, isotropic exchange is taken care of with a single orbital setting. For spin-conserving Hamiltonians with $N$ spatial orbitals encoded via Jordan-Wigner, the $X$-rank, the binary rank of the $X$ parts of its strings, of any commuting Pauli string family is no more than $2(N-1)$. Searches with the Robust Adaptive Nature-inspired Global Explorer (RANGE) reach this limit for CH$_4$, with the rank of the Hamiltonian support vouching for the maximality of the search values. And for all nonidentity Pauli strings on $n$ qubits, semidefinite witnesses from the Lovász theta number give a lower bound on the Clifford setting count that is in agreement with the integer cover and counting bounds at $2^n+1$.
\end{abstract}
\section{Introduction}
To resolve its fluctuations, an electronic observable has to be broken down into measurable fragments and sampled accordingly. A \emph{fragment} is what is drawn from a \emph{setting}, a basis change and subsequent qubit readout, and the sum of the fragments is the observable. Orbital rotations make for a convenient basis change since they let one access various fermionic occupation operators while allowing the particle number and spin populations to be preserved~\cite{Izmaylov2019,Yen2021,Choi2023}. How many \emph{shot}s (one state preparation and measurement) are required is a matter of the variances in those fragments.

We have determined the precise minimum of orbital settings for a system with one electron per spin in two spatial orbitals. The logical qubits are defined by the two possible orbital occupations of each electron, and the allowed rotations, which may be complex, are separate in the spin sectors. Any setting adds at most rank one to the coefficient matrix of the Pauli operators on both logical qubits. By factoring in the one-qubit terms we convert this rank bound into an exact figure and build a decomposition to match it, meaning the model's observables can be put back together with three settings at most.

There are cases where joint measurements are more economical than running independent orbital rotations. After a Clifford circuit has been applied, commuting Pauli strings, tensor products of identity and Pauli matrices, can be measured in concert~\cite{Verteletskyi2020,Gokhale2020,YenVerteletskyi2020,Jena2022}. Take an exchange observable with three nonzero couplings: it is diagonal in the Bell basis of maximally entangled states and while it demands three orbital settings, one Clifford setting suffices. We present the readout circuit for the four physical qubits involved. While Gühne and Hyllus made use of a similar rank argument for entanglement witnesses~\cite{GuhneHyllus2003}, we go further by including arbitrary one-qubit terms, specifying the associated fermionic rotations and deriving the optimal sampling costs.

Even with a fixed number of settings, the optimal shot ratios will vary with the state. For isotropic and anisotropic exchange we derive these ratios. In a uniform mixture of the four logical basis states, the constant for an optimal orbital measurement (where there are no one-qubit terms) is the sum of the singular values of the correlation matrix. A joint Clifford readout of a Bell-diagonal observable, by contrast, gives the square root of their squared sum. Since the leading shot cost scales with the square of the constant, the ratio for isotropic exchange comes to three. Varying the populations of the logical singlet and triplet sees this ratio tend to three halves near the logical singlet, where both leading costs vanish.

The exchange separation depends on conserving the two spin populations individually, and a particle-number-preserving rotation that mixes the spin sectors makes isotropic exchange readable in a single setting. One should not assume either class of measurement is superior in all instances. A generically rotated Pauli axis requires only one orbital setting but two from the Clifford side. Our constructions and bounds are meant to span the whole orbital class, unlike the finite collections considered in the certification work of~\cite{Zahariev2026CCF}. Variational quantum chemistry~\cite{Peruzzo2014,McClean2016} and the calculation of properties and derivatives~\cite{Mitarai2020} are in large part a matter of measurement. The standard workflow is to expand the operator in Pauli strings, put compatible terms together~\cite{Izmaylov2019,Verteletskyi2020}, handle the basis changes and allocate the necessary shots. Before simulation one may also employ orbital rotations to alter the Hamiltonian representation and bring down the coefficient norm~\cite{Koridon2021}. This is a different matter from a measurement rotation meant to lay bare a given fragment.

There is a clear division of labor between proof and search. The algebra yields the correlation-rank obstruction, the approximation formula and the singular value decomposition. One can use RANGE (robust adaptive nature-inspired global explorer)~\cite{RANGE} for an independent check on their realization by way of a continuous search over orbital angles, but the proof does not call for it. A discrete variant of RANGE is put to work on the separate task of locating commuting subsets in a molecular Hamiltonian with a high $X$-rank (the binary rank of the $X$ components of the Pauli strings). For $N$ spatial orbitals under Jordan-Wigner encoding, spin conservation puts a cap on that at $2(N-1)$. Any search value is only the best found unless it hits a proven upper bound like the parity ceiling or the rank of the entire Hamiltonian support.

We present our results as follows:
\begin{enumerate}[label=(\roman*),leftmargin=1.7em,itemsep=2pt,topsep=3pt]
\item \emph{Measurement primitives.} This includes the zero-$T$ consequence for measurement of the simultaneous Clifford diagonalization of commuting Pauli strings~\cite{YenVerteletskyi2020}. Hadamard tests with controlled-Pauli insertions fall into the category of Clifford operations.
\item \emph{Exact setting count.} In the fixed $(1,1)$ sector the augmented correlation-matrix ranks tell us how many orbital settings are needed for a general observable, even with one-qubit terms. We get the exact error in the correlation approximation from the singular-value tails, along with a stable lower bound for the full observable.
\item \emph{Physical separation and cost.} An explicit Clifford readout and commuting four-qubit representatives make up the Bell-diagonal family. For isotropic exchange there is an exact leading shot ratio that depends on the state. At maximal mixing the factor is three.
\item \emph{Support bounds.} The CH$_4$ commuting family shows that the Jordan-Wigner ceiling $r_X\leq2(N-1)$ imposed by spin conservation can be met. This is a point about support, not exclusion from the orbital class.
\item \emph{Molecular covers.} Theta witnesses put a bound on the integer grouping problem. The complete nonidentity Pauli set has an exact cover number of $2^n+1$.
\item \emph{Structural searches.} Discrete search is used to find high-$X$-rank families in everything from small molecules to 36-qubit heavy-element operators without tapering, while a continuous one verifies the orbital-angle constructions. The Hamiltonian support rank is what makes every value found maximal.
\end{enumerate}

\section{Orbital measurements in the two-electron model}\label{sec:prelim}
Here an orbital setting is used to measure occupations in a rotated basis. Since we allow for any classical processing of the results, any observable that is diagonal in the measured occupations is a fragment one can access. In a decomposition $O=\sum_{k=1}^{K}F_k$ each fragment is given to an admissible setting; $K_{\rm orbital}(O)$ is the minimum required. There is no need to sample a known identity so it is taken out. The Cartan subalgebra (CSA) method for orbital fragmentation~\cite{Yen2021} is based on these same commuting occupation operators.

Consider a case of one electron in two orbitals. With spins $\alpha$ and $\beta$ and orbitals 1 and 2, each spin sector is a two-level system. Defining
\begin{equation}
|0\rangle_\sigma=a_{1\sigma}^\dagger|\Omega\rangle,
\qquad |1\rangle_\sigma=a_{2\sigma}^\dagger|\Omega\rangle.
\label{eq:logical-sector}
\end{equation}
these basis states serve as a logical qubit for the $\sigma$ spin, with Pauli operators
\begin{align}
X_\sigma^L&=a_{1\sigma}^\dagger a_{2\sigma}+a_{2\sigma}^\dagger a_{1\sigma},\\
Y_\sigma^L&=-i(a_{1\sigma}^\dagger a_{2\sigma}-a_{2\sigma}^\dagger a_{1\sigma}),\\
Z_\sigma^L&=n_{1\sigma}-n_{2\sigma}.
\end{align}
The first pair are for coherent orbital transfer, the third for the population difference.

Arbitrary local rotations of the logical qubits are carried out by independent orbital rotations in the two spin sectors. These are particle-number-preserving fermionic Gaussian unitaries coming from quadratic operators. While one might apply the same real spatial-orbital rotation to both spins, we allow a complex rotation for each, $u_\alpha\otimes u_\beta$, in the $(N_\alpha,N_\beta)=(1,1)$ sector. The class is wider than the single shared rotation of the orbital-basis study and the bounds here should not be conflated with those of the narrower class. Spin mixing is left to Sec.~\ref{sec:physical} and Bogoliubov transformations that do not conserve number are out. The elementary Givens rotations of the Jordan-Wigner encoding, which put an orbital rotation into effect, possess continuous angles. As a rule they require non-Clifford synthesis in order to be fault-tolerant to the necessary precision~\cite{RossSelinger2016}.

With each setting we measure a local component of every logical qubit and the correlation between them. In the unrotated basis one can access operators that are linear combinations of $II$, $ZI$, $IZ$ and $ZZ$. A single-qubit rotation will map $Z$ onto $u\cdot\sigma=u_xX+u_yY+u_zZ$ for some real unit vector $u$. An arbitrary traceless fragment then takes the form
\begin{equation}
\begin{split}
F={}&\alpha(u\cdot\sigma)\otimes I+
\beta I\otimes(v\cdot\sigma)\\
&+\gamma(u\cdot\sigma)\otimes(v\cdot\sigma),
\end{split}
\label{eq:one-context}
\end{equation}
with $\alpha,\beta,\gamma$ as real coefficients and $u,v$ being independent unit vectors.

Readout in the Clifford group is a joint measurement of commuting Pauli strings on the physical qubits. The group is generated by Hadamard ($H$), phase ($S$) and controlled-NOT (CNOT) gates. The $S$ gate, for instance, puts a $\pi/2$ phase on $|1\rangle$.
\begin{theorem}[Simultaneous diagonalization by a Clifford circuit]
\label{thm:clifford-diag}
Given a set $S=\{P_1,\ldots,P_m\}$ of $n$-qubit Pauli strings that all commute with one another, there is a Clifford unitary $U$ so that $UP_iU^\dagger$ is nothing but a string of $I$ and $Z$ operators (up to sign). The family can be read out once a basis change has been effected with $H$, $S$ and CNOTs.
\end{theorem}
This is a standard stabilizer result~\cite{Gottesman1997,AaronsonGottesman2004} and the foundation of fully commuting grouping~\cite{YenVerteletskyi2020}, though hardware-aware versions will respect connectivity~\cite{Miller2024}. We go over the proof for clarity on the resource implications.
\begin{proof}
As shown in the binary representation of Appendix~\ref{app:symplectic}, two strings commute if and only if the symplectic form of their vectors is zero. The vectors of $S$ therefore span a subspace $W$ where the form vanishes. One can extend a basis of $W$ to a maximal subspace of dimension $n$ and thence to a symplectic basis of the whole space. Under a symplectic transformation the maximal subspace is mapped to the standard $Z$ space of the vectors $(0\mid e_j)$ for $j=1,\ldots,n$. Any such binary transformation admits a Clifford implementation via $H$, $S$ and CNOT~\cite{Gottesman1997,AaronsonGottesman2004} (modulo Pauli phases), turning the original string into a signed $Z$ string. For the constructive elimination see Appendix~\ref{app:symplectic}. The four-qubit version we employ is in Sec.~\ref{sec:physical}.
\end{proof}
\begin{corollary}[No measurement-basis $T$ gates]
\label{cor:zero-T}
The basis changes in Theorem~\ref{thm:clifford-diag} have a $T$ count of zero (a $T$ gate would apply a $\pi/4$ phase to $|1\rangle$). That does not mean there is no cost in sampling or state preparation, nor do the Clifford gates come free. Controlled-Pauli insertions are left to Appendix~\ref{app:controlled}. There is no requirement here for particle-number conservation or a fermionic-fragment construction.
\end{corollary}
The result needs neither a fermionic-fragment construction nor particle-number conservation.
\section{The exact number of orbital settings}\label{sec:separation}
\label{sec:full-count}
We consider the Hilbert space of the fixed $(1,1)$ sector for two spatial orbitals per spin in what follows. Each spin sector has a logical qubit, and the allowed orbital rotations are independent for the two sectors. Our aim is to pin down the rank of a single setting and from there the total number of settings required for an observable containing both one-qubit and correlation terms.

\subsection{Correlation rank of a single setting}
Take a two-qubit observable $O$. Its correlation matrix is given by
\begin{equation}
T(O)_{ij}=\frac14\Tr\bigl[O(\sigma_i\otimes\sigma_j)\bigr],
\quad i,j\in\{x,y,z\}.
\label{eq:Tcorr}
\end{equation}
\begin{lemma}[A single setting contributes rank at most one]
\label{lem:rankone}
Let us look at a fragment that results from the application of an allowed product rotation $u_\alpha\otimes u_\beta$ to
\[
D=d_0I\otimes I+d_\alpha Z\otimes I+d_\beta I\otimes Z
+d_{\alpha\beta}Z\otimes Z.
\]
The correlation matrix for this will be
\[
T\bigl[(u_\alpha\otimes u_\beta)^\dagger D(u_\alpha\otimes u_\beta)\bigr]
=d_{\alpha\beta}ab^\top,
\qquad\rank T\leq1,
\]
in which $a,b\in\mathbb R^3$ are unit Bloch vectors defined by $u_\alpha^\dagger Zu_\alpha=a\cdot\sigma$ and $u_\beta^\dagger Zu_\beta=b\cdot\sigma$.
\end{lemma}
\begin{proof}
With a one-qubit rotation we have $u^\dagger Zu=a\cdot\sigma$ and $\|a\|_2=1$. The local terms and the identity do not figure in the correlation matrix. What is left is $d_{\alpha\beta}(a\cdot\sigma)\otimes(b\cdot\sigma)$, the coefficients of which make up the outer product in question.
\end{proof}
\begin{theorem}[Correlation-rank lower bound]
\label{thm:rank}
For the fixed $(1,1)$ model, reconstructing $O$ while holding each spin population constant demands a minimum number of orbital settings satisfying
\begin{equation}
K_{\rm orbital}(O)\geq\rank T(O).
\label{eq:rankbound}
\end{equation}
\end{theorem}
\begin{proof}
If we decompose $O$ as $\sum_{k=1}^KO_k$ then $T(O)=\sum_kT(O_k)$. Lemma~\ref{lem:rankone} puts a cap of one on the rank of any such summand, so rank subadditivity tells us $\rank T(O)\leq K$.
\end{proof}

\subsection{The exact count including one-qubit terms}
Once the one-qubit terms are taken into account, the correlation-rank bound serves as an exact criterion for reconstruction. One can write a traceless observable in the form
\begin{equation}
O=\sum_i a_i\sigma_i\otimes I+\sum_j b_j I\otimes\sigma_j
  +\sum_{ij}T_{ij}\sigma_i\otimes\sigma_j.
\label{eq:full-observable}
\end{equation}
with the real $3\times3$ matrix $T$ gathering the coefficients that couple the qubits (indexed by $X$, $Y$, $Z$) and $a$ and $b$ the coefficients of the logical-qubit operators. To get the exact count, append the local coefficient vector to the appropriate correlation matrix.

\begin{theorem}[Exact setting count]
\label{thm:full-count}
Under the fixed $(1,1)$ model, independent orbital rotations that leave the spin populations intact call for precisely
\begin{equation}
\boxed{K_{\rm orbital}(O)=\max\{\rank[T\ a],\rank[T^\top\ b]\}}
\label{eq:full-context-count}
\end{equation}
settings. This is $r=\rank T$ provided $a$ is in the column space of $T$ and $b$ in that of $T^\top$. Otherwise it is $r+1$. The zero observable is an exception requiring none.
\end{theorem}
\begin{proof}
Eq.~\eqref{eq:one-context} describes a setting that provides a pair of axes $u_k,v_k$. For the observable to be reconstructed these must, along with their coefficients, meet the condition
\begin{align*}
T&=\sum_k\gamma_k u_kv_k^\top,\\
a&=\sum_k\alpha_k u_k,\qquad b=\sum_k\beta_k v_k.
\end{align*}
Every column of $[T\ a]$ is in the span of the $K$ $u_k$ vectors, as is every column of $[T^\top\ b]$ in the span of the $v_k$. It follows that the augmented matrices have a rank no higher than $K$.

We can reach the bound via an SVD. Setting $T=\sum_{i=1}^r s_i u_i v_i^\top$, we take the $r$ singular-vector pairs for our measurement axes. For the $i$-th setting the correlation coefficient is $s_i$ and the local ones $u_i^\top a$ and $v_i^\top b$. Any orthogonal residuals $a_\perp$ and $b_\perp$ are handled in a further setting with zero correlation coefficient and axes aligned to the nonzero residual. Where there is no residual the axis is arbitrary. In this way the construction conforms to the count in Eq.~\eqref{eq:full-context-count}.
\end{proof}

Three orbital settings are enough for any observable in the model. A rank-three $T$ has local vectors in its row and column spaces already, and a matrix of lower rank needs at most one more. Take $XX+ZI$: the local $Z$ is not in the correlation direction $X$, so two settings are called for. Should the number of settings be limited, the augmented matrices put constraints on the reconstruction error.

\begin{corollary}[A stable bound on reconstruction error]
\label{cor:stable-count}
We define the coefficient norm as $\|O\|_{\rm c}^2=\|a\|_2^2+\|b\|_2^2+\|T\|_F^2$, the Frobenius norm being the square root of the sum of squared entries. For some matrix $M$ let $\tau_K(M)^2=\sum_{i>K}\sigma_i(M)^2$ be what is left of the squared singular values after the top $K$ are gone. Then any $\widetilde O$ put together from $K$ permitted settings will satisfy
\begin{equation}
\|O-\widetilde O\|_{\rm c}\geq
\max\{\tau_K([T\ a]),\tau_K([T^\top\ b])\}.
\label{eq:stable-full-count}
\end{equation}
If the estimate is within $\delta$ of the target in this norm, subtracting $\delta$ from the right-hand side (and taking zero if it goes negative) gives a valid bound on the true target.
\end{corollary}
\begin{proof}
The Eckart-Young theorem means the singular values we have discarded set a lower limit on the approximation error for each of the augmented matrices (which are of rank at most $K$). Since the error in any one matrix cannot exceed the total coefficient error, Eq.~\eqref{eq:stable-full-count} is proved. Shifting the target by a coefficient distance $\delta$ alters its proximity to the class of admissible approximations by no more than $\delta$.
\end{proof}

We have put the construction to the test with 1300 random targets of all correlation ranks from zero to three and found the predicted count holds, with coefficient residuals under $6.5\times10^{-15}$ even when local terms are outside the singular subspaces. Equation~\eqref{eq:stable-full-count} is verified for another 500 perturbed targets with known decompositions.

\subsection{Best approximation of the correlation matrix}
When dealing with the correlation matrix alone, the singular values that are thrown out represent the minimum approximation error exactly. This is because an orbital fragment is capable of producing any real rank-one correlation matrix.
\begin{theorem}[Exact correlation approximation error]
\label{thm:eckart}
Let $\sigma_1\geq\sigma_2\geq\sigma_3\geq0$ be the singular values of $T(O)$. Then the least correlation error across the $T_k$ contributions of $K$ orbital settings is
\begin{equation}
\min\left\|T(O)-\sum_{k=1}^KT_k\right\|_F
=\left(\sum_{i>K}\sigma_i^2\right)^{1/2}.
\label{eq:eckart}
\end{equation}
So $K\geq\rank T(O)$ is required and sufficient to get an exact reconstruction of the correlation part.
\end{theorem}
\begin{proof}
An arbitrary setting yields $T_k=d_k a_kb_k^\top$ for arbitrary real unit vectors and coefficient $d_k$. The range of possible contributions is precisely the set of real $3\times3$ matrices of rank one or less. All matrices of rank no greater than $K$ are the result of summing $K$ such contributions. The Eckart-Young minimum error~\cite{EckartYoung1936} is obtained by way of an SVD truncation, with orbital rotations providing the necessary singular-vector axes.
\end{proof}

One can see the outcome plainly in Bell-diagonal observables since their local vectors are zero. Take for instance the couplings $(0.7,-0.4,1.1)$. These yield singular values of $(1.1,0.7,0.4)$ and leave minimum residuals of $0.806$, $0.400$ and nothing at all for one, two or three orbital settings. The Clifford construction that follows measures this observable exactly in a single setting. Where there are nonzero local terms, Corollary~\ref{cor:stable-count} puts a floor on the total error as opposed to giving the full optimal approximation curve.

It is worth looking at rotation-angle coordinates to verify. Because a logical qubit rotation will cover every unit vector in $\mathbb R^3$, the truncated-SVD configurations are physical orbital measurements and not simply a matter of matrix factorization. An independent RANGE search~\cite{RANGE} in these coordinates comes back with the expected error to six figures when $K=1$ and four when $K=2$. For $K=3$ the archived finite-budget figure is $0.400$, identical to the $K=2$ residual. The third useful setting was not found in that run. The SVD approach has no residual, so any discrepancy is down to an incomplete search rather than imprecision in the construction.

\subsection{More orbitals per spin}
\begin{theorem}[A correlation-rank bound for $m$ orbitals per spin]\label{theo:qudit}
Theorem 5 provides a correlation-rank bound for m orbitals per spin. Put one electron of each spin in $m\geq2$ spatial orbitals per spin, meaning $m$ levels for each orbital degree of freedom. If $O$ is expanded in a product operator basis and the coefficients of traceless pairs are gathered in $T(O)$, the matrix is of dimension $(m^2-1)\times(m^2-1)$. One orbital setting can add at most rank $m-1$. In the case of m = 2 we have the lower bound of Theorem 2.
\[
K_{\rm orbital}(O)\geq\left\lceil\frac{\rank T(O)}{m-1}\right\rceil.
\]
For $m=2$, this reduces to the correlation-rank lower bound in Theorem~\ref{thm:rank}.
\end{theorem}
\begin{proof}
$D$ is diagonal in the occupation basis, and a valid setting measures $(u_\alpha\otimes u_\beta)^\dagger D(u_\alpha\otimes u_\beta)$. The correlation part may be expressed as $\sum_{a,b}c_{ab}E_a\otimes E_b$ using $m-1$ traceless diagonal basis matrices $E_1,\ldots,E_{m-1}$. Independent rotations keep the coordinate vectors of these within a span of at most $m-1$ on both sides, so the contribution from any given setting is a correlation matrix of rank at most $m-1$. Summing over $K_{\rm orbital}$ fragments means $\rank T(O)\leq K_{\rm orbital}(m-1)$.
\end{proof}

With more than two orbitals per spin, however, rank alone does not tell the whole story. The $m-1$ operator directions have to come from a rotated diagonal subspace and are not free to be chosen at will. Some matrices of rank at most $m-1$ may need more than one setting. The statement for the larger space is thus a lower bound, while the exact classification of Theorem~\ref{thm:full-count} is for two orbitals per spin only.
\section{Exchange measurements with orbital rotations and Clifford readout}
\label{sec:physical}
An exchange involving three nonzero Bell-diagonal couplings calls for three orbital settings yet only a single Clifford one. The observable is
\begin{equation}
O_{\rm Bell}=\lambda_xX_LX_L+\lambda_yY_LY_L+\lambda_zZ_LZ_L,
\label{eq:bellfamily}
\end{equation}
the factors acting on the two spin sectors. Isotropic exchange is the case of equal couplings. When none of the three are zero the diagonal correlation matrix is of rank three and the Pauli products commute, yielding
\begin{equation}
K_{\rm orbital}(O_{\rm Bell})=3,
\qquad K_{\rm Clifford}(O_{\rm Bell})=1.
\label{eq:contextseparation}
\end{equation}
Here $K_{\rm Clifford}(O)$ is the corresponding minimum under Clifford settings. The orbital count is reached with settings on the $X$, $Y$ and $Z$ axes.

Four physical qubits can be used to implement the joint measurement. Let logical qubits $A$ and $B$ be encoded in the spin-orbital pairs $(1,2)$ and $(3,4)$, one electron to a pair. We define the code via parity operators $S_A=Z_1Z_2$ and $S_B=Z_3Z_4$ which have eigenvalue $-1$ for these occupations. For the first logical qubit we might take as representatives, sign aside,
\begin{equation}
X_L\sim X_1X_2,\qquad Y_L\sim Y_1X_2,\qquad Z_L\sim Z_1,
\label{eq:logicalreps}
\end{equation}
and similarly for the second pair. The physical exchange strings are
\begin{align}
P_X&=X_1X_2X_3X_4,\\
P_Y&=Y_1X_2Y_3X_4,\\
P_Z&=Z_1Z_3.
\end{align}
These commute with one another and with $S_A,S_B$ and satisfy $P_Y=-P_XP_Z$. The circuit
\begin{equation}
U_{\rm phys}=H_1\,\mathrm{CNOT}_{1\to4}\,
\mathrm{CNOT}_{1\to2}\,\mathrm{CNOT}_{1\to3},
\label{eq:uphys}
\end{equation}
with the rightmost gate going first, conjugates to give
\begin{align}
U_{\rm phys}P_XU_{\rm phys}^\dagger&=Z_1,\\
U_{\rm phys}P_ZU_{\rm phys}^\dagger&=Z_3,\\
U_{\rm phys}P_YU_{\rm phys}^\dagger&=-Z_1Z_3.
\end{align}
One readout of the four qubits is enough for all three exchange terms. Our validation scripts check the commutation relations and operator identities, as well as the integrity of the encoded sector. They also run through over 300 orbital settings to bear out rank subadditivity up to $K=3$.

\begin{remark}[Spin mixing removes the isotropic-exchange restriction]
\label{rem:spinmix}
The restriction on isotropic exchange is lifted if spin mixing is allowed. Rotations can make it possible in a single orbital setting. For mode order $(1\alpha,2\alpha,1\beta,2\beta)$ let $b_{p,\pm}^\dagger=(a_{p\alpha}^\dagger\pm a_{p\beta}^\dagger)/\sqrt2$ for $p=1,2$. This preserves the total number of particles but not the spin populations individually. Restricted to the (1,1) sector the operator
\[
D=I-4(n_{1+}n_{2+}+n_{1-}n_{2-})
\]
is nothing more than $X_\alpha^LX_\beta^L+Y_\alpha^LY_\beta^L+Z_\alpha^LZ_\beta^L$ in its occupation basis, a fact confirmed to machine precision. The need for three settings in an orbital readout of this observable is therefore contingent on conserving each spin population separately.
\end{remark}

\begin{remark}[Neither measurement class contains the other]
\label{rem:incomparable}
There is no inclusion between the two classes of measurement. A continuously rotated orbital axis can be preferred for orbital readout. In a two-orbital, one-particle sector $O_\theta=\cos\theta\,Z_L+\sin\theta\,X_L$ (or $O_\theta\otimes I$ in the two-electron model) takes one orbital setting but generically two Clifford Pauli settings. The Pauli components do not commute even though the observable is fully diagonal in a rotated orbital basis. Together with the Bell-family case, this makes it clear that one class does not subsume the other, and there is value in making use of both.
\end{remark}

\section{Optimal sampling costs}\label{sec:cost}

With a reconstruction in place, the fragment variances set the sampling cost. Consider a state $\rho$ and per-shot costs $c_k$ that are positive. Let $O=\sum_kF_k$ and $v_k=\Var_\rho(F_k)$. If one works with independent batches, the total variance is $\sum_kv_k/N_k$ for an outlay of $\sum_kc_kN_k$. To find the shot counts that yield the best leading cost at standard error $\varepsilon$, we invoke the Cauchy-Schwarz inequality:
\begin{equation}
\begin{aligned}
C(\varepsilon)&=\frac{\Phi_{\rho,c}(O)^2}{\varepsilon^2},\\
\Phi_{\rho,c}(O)&=\inf\sum_k\sqrt{c_k\Var_\rho(F_k)}.
\end{aligned}
\label{eq:certifiedcost}
\end{equation}
The expression above is what is proportional to $\varepsilon^{-2}$ in the limit as $\varepsilon$ goes to zero, the infimum being over any finite reconstructions the settings permit. We term $\Phi_{\rho,c}(O)$ the measurement constant. The optimal way to apportion shots for given fragments is to have $N_k$ scale with $\sqrt{v_k/c_k}$. Rounding up incurs only lower-order expense, even if some deterministic fragment calls for a preliminary readout. There can be sharing of coefficients between settings, as seen in overlapping measurement optimization~\cite{Crawford2021,Yen2023}. The related certification work~\cite{Zahariev2026CCF} puts bounds on this objective for finite setting collections.

A general lower bound is met by measuring the whole observable in a single setting when per-shot costs are unity. The triangle inequality for standard deviations ensures $\sum_k\sqrt{\Var_\rho(F_k)}\geq\sqrt{\Var_\rho(O)}$ for any decomposition, with equality holding for a direct readout of $O$. In what follows we compare using equal unit shot costs for both classes of measurement, the specified density matrix and eigenvalue scores.

One may also include non-Clifford gates in the per-shot tally. Say the $k$th basis change involves $T_{\rm meas}^{(k)}$ of type $T$ and the state preparation $T_{\rm prep}$ of them, and that $M_k$ shots are put into setting $k$. The non-Clifford cost then comes to
\begin{equation}
C_T=\sum_{k=1}^K M_k\bigl(T_{\rm prep}+T_{\rm meas}^{(k)}\bigr),
\label{eq:ct-general}
\end{equation}
since payment is only for the setting in execution. Setting $c_k=T_{\rm prep}+T_{\rm meas}^{(k)}$ in Eq.~\eqref{eq:certifiedcost} demands a positive sum. Should all $T$ counts be nil, the resource measure is zero and fails to capture the rest of the sampling or execution load.

As noted in Remark~\ref{rem:incomparable}, the two single-setting classes do not allow for comparison. It is not the case for every observable that $\Phi_{\rm Clifford}\leq\Phi_{\rm orbital}$ or vice versa. Let $\mathcal D$ be a set of settings and $\Phi_{\mathcal D}$ the measurement constant optimized for decompositions into fragments readable in $\mathcal D$. For the union of the orbital and Clifford sets, $\mathcal D_G\cup\mathcal D_C$, we have
\[
\Phi_{\mathcal D_G\cup\mathcal D_C}
\leq\min\{\Phi_{\mathcal D_G},\Phi_{\mathcal D_C}\},
\]
by virtue of either original decomposition being open. While correlation rank dictates what orbital decompositions will reconstruct a target, the variances and per-shot costs are what give statistical weight to such constraints.

\subsection{Costs at maximal mixing}
Here the singular values provide closed forms for the optimal constants. Their sum is the nuclear norm $\|T\|_*$ and the square root of the sum of squares is the Frobenius norm $\|T\|_F$.

\begin{proposition}[Exact cost at maximal mixing]
\label{prop:price}
Unit variance and zero covariance characterize distinct nonidentity logical Pauli products for the maximally mixed state in the four-dimensional logical sector, assuming unit per-shot costs. Where the observable has only a two-body block $T(O)$, any admissible decomposition has a measurement constant no less than $\|T(O)\|_F$. Over the full orbital class the optimum is
\[
\Phi_{\rm orbital}(O)=\|T(O)\|_*.
\]
If a single Clifford setting reads $O$, as it does for any Bell-diagonal observable, then $\Phi_{\rm Clifford}(O)=\|T(O)\|_F$. This makes the measurement constants for isotropic Heisenberg exchange $O=J(XX+YY+ZZ)$ come to $3|J|$ and $\sqrt3|J|$ respectively, a factor of three in the leading sampling cost. Of the Bell-diagonal observables with a set correlation rank, the ratio is at its highest when the nonzero couplings are of like magnitude.
\end{proposition}
\begin{proof}
As a seminorm, standard deviation means any $O=\sum_kF_k$ must satisfy $\sum_k\sqrt{\Var(F_k)}\geq\sqrt{\Var(O)}=\|T(O)\|_F$. A physical Bell readout or a setting that measures the observable in full will hit this mark.

An orbital fragment brings $d_k(a_k\cdot\sigma)\otimes(b_k\cdot\sigma)$ to the correlation side along with local or identity terms as the case may be. At maximal mixing there is no covariance between distinct logical Pauli products. Stripping away the local terms from each fragment does not alter the reconstruction of a target lacking them, nor does it raise the variance. What is left are standard deviations of $|d_k|$ where $\sum_kd_ka_kb_k^\top=T(O)$.

Take a real matrix $W$ whose largest singular value is at most one. The pairing $\langle T,W\rangle=\operatorname{Tr}(T^\top W)$ obeys
\[
\langle T(O),W\rangle=\sum_kd_k a_k^\top Wb_k\leq\sum_k|d_k|.
\]
By choosing $W=\sum_i u_iv_i^\top$ across the nonzero singular directions the left side becomes $\|T(O)\|_*$. The SVD fragments with $d_i=\sigma_i$, $a_i=u_i$ and $b_i=v_i$ reach the bound, making $W$ a witness for the lower bound in the associated optimisation~\cite{Zahariev2026CCF}.

If a Clifford setting gets to the Frobenius bound, the shot cost ratio is $(\sum_i\sigma_i)^2/\sum_i\sigma_i^2$, which is bounded by $\rank T(O)$. Equality of the nonzero singular values is required to attain the upper end. With isotropic exchange the three singular values are $|J|$, giving $(3/\sqrt3)^2=3$. These analytic results are borne out by an optimization of 600 random orbital settings and those on the coordinate axes, the lower-bound proof and the attaining axis construction proving optimality over the full orbital class.
\end{proof}

\subsection{State-dependent cost of isotropic exchange}
\label{sec:state-price}
The optimal shot ratio depends on the relative logical singlet and logical triplet populations. We turn to the isotropic family of logical states
\begin{equation}
\rho_w=\tfrac14[I-w(XX+YY+ZZ)],\qquad-\tfrac13\leq w\leq1.
\label{eq:werner-state}
\end{equation}
Here the logical singlet is found with a population of $(1+3w)/4$, while each logical triplet has $(1-w)/4$, which accounts for the positivity interval shown. The pure logical singlet is at $w=1$. Maximal mixing is seen at $w=0$.

\begin{proposition}[Optimal exchange costs along the isotropic family]
\label{prop:werner-cost}
Taking $O=J(XX+YY+ZZ)$ where $J\ne0$, and assuming unit per-shot costs for joint Clifford readout and independent spin-preserving orbital rotations, the measurement constants are as follows:
\begin{align}
\Phi_{\rm orbital}&=3|J|\sqrt{1-w^2},\\
\Phi_{\rm Clifford}&=|J|\sqrt{3(1-w)(1+3w)}.
\end{align}
In the range $-1/3<w<1$, the leading cost ratio is
\begin{equation}
\frac{C_{\rm orbital}}{C_{\rm Clifford}}
=\frac{3(1+w)}{1+3w}.
\label{eq:werner-ratio}
\end{equation}
\end{proposition}
\begin{proof}
Let $u$ and $v$ be unit vectors and put $t=u^\top v$. With $A=(u\cdot\sigma)\otimes I$ and $B=I\otimes(v\cdot\sigma)$, the local means vanish and we have $\langle AB\rangle=-wt$. The variance of any fragment $F=\alpha A+\beta B+\gamma AB$ is then
\[
\Var_{\rho_w}(F)=\alpha^2+\beta^2-2wt\alpha\beta
                 +\gamma^2(1-w^2t^2).
\]
Because the local quadratic form is positive semidefinite, $\sqrt{\Var(F)}$ is no less than $\sqrt{1-w^2}|\gamma|$. To reconstruct $JI_3$ one needs $\sum_k|\gamma_k|\geq3|J|$ in accordance with the nuclear-norm bound, a lower limit for orbital measurement that is reached by three settings with their coordinate axes aligned.

The Bell setting can measure $O$ in one readout. If $S=XX+YY+ZZ$, then from the identity $S^2=3I-2S$ it follows that $\Var(JS)=3J^2(1-w)(1+3w)$. This standard deviation is the floor for any decomposition into fragments under the seminorm, and the Bell measurement achieves it.
\end{proof}

As Fig.~\ref{fig:exchange-ratio} shows, the cost ratio drops from three when $w=0$ to $1.8$ at $w=1/2$, approaching $3/2$ as $w$ nears one. At the logical singlet end point the ratio is not defined since both leading costs are zero (Fig.~\ref{fig:exchange-costs}). For $w=-1/3$ the state is an even mixture of logical triplets and exchange has a uniform eigenvalue over its support. The Bell readout is deterministic even if the orbital fragments measured on their own show some fluctuation. We have verified these formulas with direct density-matrix work at seven interior points and 700 witness checks with random settings.

\subsection{Anisotropic exchange}

One can also derive exact optimal costs for this state family in the case of unequal exchange couplings.

\begin{proposition}[Anisotropic exchange on the isotropic family]\label{prop:aniso}
Let $O=\sum_{i=1}^{3}\lambda_{i}\,\sigma_{i}\otimes\sigma_{i}$ with real $\lambda_1,\lambda_2,\lambda_3$. In the state $\rho_{w}$ the optimal measurement constants for the given orbital class and unit per-shot costs are
\begin{align*}
\Phi_{\mathrm{orbital}}&=\sqrt{1-w^{2}}\,\sum_{i}|\lambda_{i}|,\\
\Phi_{\mathrm{Clifford}}&=\Bigl[\sum_{i}\lambda_{i}^{2}+2w\sum_{i<j}\lambda_{i}\lambda_{j}\\
&\qquad-w^{2}\bigl(\textstyle\sum_{i}\lambda_{i}\bigr)^{2}\Bigr]^{1/2}.
\end{align*}
\end{proposition}

\begin{proof}
The reasoning of Proposition~\ref{prop:werner-cost} puts a bound on each orbital fragment of $\sqrt{\operatorname{Var}_{\rho_{w}}(F)}\ge\sqrt{1-w^{2}}\,|\gamma|$. The nuclear norm of $\operatorname{diag}(\lambda)$ being $\sum_{i}|\lambda_{i}|$, the three coordinate-aligned settings will hit the orbital lower bound. As for Clifford readout, note that $(\sigma_{i}\otimes\sigma_{i})(\sigma_{j}\otimes\sigma_{j})=-\sigma_{k}\otimes\sigma_{k}$ for $i\neq j$ ($k$ being the other index), so $O^{2}=(\sum_{i}\lambda_{i}^{2})\,I-2\sum_{i<j}\lambda_{i}\lambda_{j}\,\sigma_{k}\otimes\sigma_{k}$. Plugging in $\langle O\rangle_{\rho_{w}}=-w\sum_{i}\lambda_{i}$ and $\langle\sigma_{k}\otimes\sigma_{k}\rangle_{\rho_{w}}=-w$ produces the variance in question. The Bell readout meets the seminorm lower bound as in Proposition~\ref{prop:werner-cost}.
\end{proof}

Setting all $\lambda_i=J$ brings us back to Proposition~\ref{prop:werner-cost}. With unequal couplings the shot ratio is a function of $w$ and the anisotropy. At $w=0$ a nonzero target yields $\Phi_{\mathrm{orbital}}^2/\Phi_{\mathrm{Clifford}}^2=(\sum_i|\lambda_i|)^2/\sum_i\lambda_i^2$ in line with Proposition~\ref{prop:price}. Our scripts confirm these results against 50 random $(\lambda,w)$ pairs via direct density-matrix calculation.

\begin{figure}[tbp]\centering
\includegraphics[width=\columnwidth]{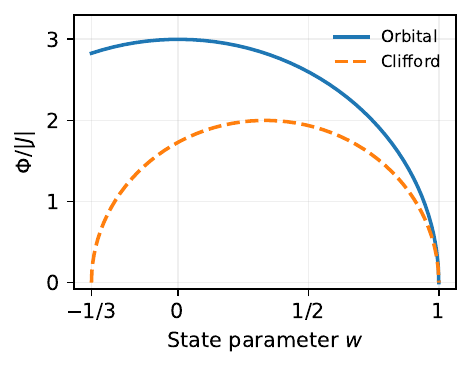}
\caption{Exact measurement constants for isotropic exchange $O=J(XX+YY+ZZ)$ in $\rho_w=[I-w(XX+YY+ZZ)]/4$, with unit cost per shot. Products such as $XX$ act on the logical qubits, and $J$ denotes the exchange coupling. The curves show $\Phi/|J|$ for independent orbital rotations that preserve each spin population and for joint Clifford readout, as given by Proposition~\ref{prop:werner-cost}. At standard error $\varepsilon$, the leading shot count is $\Phi^2/\varepsilon^2$. Both constants are zero at the logical singlet $w=1$, whereas only the Clifford constant is zero for the equal logical triplet mixture at $w=-1/3$.}\label{fig:exchange-costs}
\end{figure}
\begin{figure}[tbp]\centering
\includegraphics[width=\columnwidth]{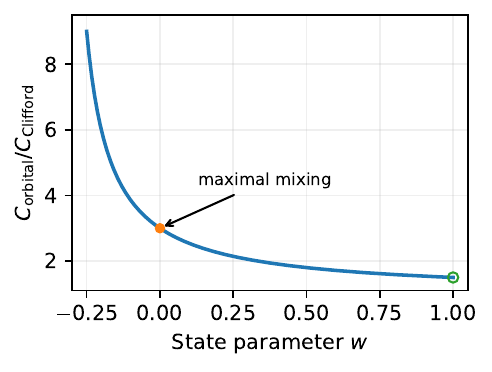}
\caption{Ratio of optimal orbital and Clifford leading shot costs for isotropic exchange $J(XX+YY+ZZ)$ in $\rho_w=[I-w(XX+YY+ZZ)]/4$. Per-shot costs are equal, and orbital rotations conserve each spin population. The plotted ratio, $3(1+w)/(1+3w)$, covers $-0.25\leq w<1$. It is three at maximal mixing and approaches $3/2$ near the logical singlet. At $w=1$, the ratio is undefined because both costs vanish. Its divergence at $w=-1/3$ is outside the plotted range.}\label{fig:exchange-ratio}
\end{figure}

\section{Spin parity and the rank of Pauli support}
\label{sec:parity}
The correlation-rank logic of Sec.~\ref{sec:separation} gives an obstruction for orbital settings in the two-orbital model. With larger Hamiltonians we look at something else: the binary rank of the Pauli $X$ support for a commuting family. In a Jordan-Wigner Givens rotation the intervening modes put in $Z$ factors while the rotated pair has $X$ and $Y$. The support is a record of off-diagonal Pauli action and is frame dependent. It does not in itself determine whether an orbital setting can read a fragment.

Following Appendix~\ref{app:symplectic}, an $n$-qubit Pauli string (to within phase) is written as the binary vector $(x\mid z)$, $x$ for the $X$ or $Y$ positions and $z$ for $Z$ or $Y$.
\begin{definition}[Binary $X$-support rank]
Given a set $S=\{P_1,\ldots,P_m\}$ of Pauli strings with vectors $(x_i\mid z_i)$, define
\begin{equation}
r_X(S)=\rank_{\F_2}\begin{pmatrix}x_1\\ \vdots\\ x_m\end{pmatrix},
\label{eq:xrank}
\end{equation}
$r_X(S)$ is the rank over the binary field of the matrix formed from the $X$ parts. We call $r_X(S)$ the $X$-rank of $S$.
\end{definition}
\begin{definition}[Orbital accessibility of support]
A commuting set $S$ is accessible in the support sense if there is an occupation-diagonal fermionic fragment $D$ and a number-preserving Gaussian unitary $U_G$ (one that leaves the spin populations alone) such that $S$ is in the Pauli expansion of $U_GDU_G^\dagger$.
\end{definition}
This tells us what strings can appear in a rotated fragment but not whether every Pauli component is readable on its own in that setting, since a commutative observable algebra may have noncommuting Pauli components (see Remark~\ref{rem:incomparable}). An $X$ block of $N$ columns has a rank of at most $N$. Conserving the two spin populations separately takes away another binary direction in each sector.

\begin{theorem}[Two-sector parity bound]
\label{thm:parity}
Consider a spin-conserving molecular electronic Hamiltonian with $N$ spatial orbitals and its Jordan-Wigner encoding in an interleaved spin-orbital sequence. For any subset $S$ of the Pauli support, the following holds:
\begin{equation}
r_X(S)\leq2(N-1).
\label{eq:parity-ceiling}
\end{equation}
A mutually commuting part of the molecular Hamiltonian is able to reach this bound.
\end{theorem}
\begin{proof}
We first look at the parity in a sector. Let a mode index be $\sigma=2p+s$, with $p=0,\ldots,N-1$ for the spatial orbital and $s=0,1$ for the spin. The even and odd qubit indices are taken by the two sectors. An operator for creation or annihilation will put down an $X$ or $Y$ factor at its mode. The Jordan-Wigner string contributes nothing but $Z$ factors. To conserve number in a given sector there must be an even count of such operators. Any repeated action at a single mode cancels out modulo two, leaving every Pauli term with an even $X$ weight per sector.

As for the dimension of the projection, the binary vectors of length $N$ with even weight make up a subspace of dimension $N-1$, so the sector projections satisfy
\begin{equation}
r_X^\alpha(S)\leq N-1,\qquad r_X^\beta(S)\leq N-1.
\end{equation}
When one considers the combined support of the full block $[X_\alpha\mid X_\beta]$, rank subadditivity yields
\begin{equation}
r_X(S)\leq r_X^\alpha(S)+r_X^\beta(S)\leq2(N-1).
\end{equation}
To demonstrate tightness we have CH$_4$ in the minimal STO-3G basis (three Gaussians to a Slater-type orbital). With $N=9$ and canonical restricted Hartree-Fock (RHF) orbitals computed without point-group symmetry, thus mixing the degenerate $t_2$ orbitals, the Hamiltonian has a commuting subset of rank 16, which is the bound. A check of the pairwise commutation and binary rank for this subset provides an attaining witness for the Sec.~\ref{sec:search} results, as Table~\ref{tab:two-sector} shows. If symmetry-adapted orbitals are employed the rank of the CH$_4$ operator's entire support is 14. The bound is unaffected but attainment is gauge dependent.
\end{proof}

\begin{remark}[On what a support bound does not exclude]
\label{rem:support-bound}
One cannot rule out a pair of orbital rotations simply by exceeding the dimension of a sector or coming close to the combined ceiling. This theorem is meant for all subsets of the Hamiltonian support, not just the images of particular orbital settings. While a stronger criterion could be based on the incidence structure of orbital pairs, the $X$-rank in itself is insufficient for exclusion. The paper's strict orbital obstruction is the correlation-rank finding of Sec.~\ref{sec:separation}.
\end{remark}

\subsection{Molecular sector checks}
Table~\ref{tab:two-sector} gives twelve molecule-basis pairings covering 4 to 26 qubits and as many as 12,731 nonidentity Pauli terms. The Hamiltonians are set up in the interleaved Jordan-Wigner order using canonical RHF orbitals from the full space, with no point-group symmetry imposed. All terms conform to the sector-parity rule. Since $r_X(S)\leq r_X(H)$ for any support subset $S$, the rank of the total support vouches for the maximality of a commuting family that reaches it, and every family in the table does.
\begin{table*}[!t]\centering\footnotesize\setlength{\tabcolsep}{5pt}
\caption{Sector-parity checks and commuting-family ranks for twelve molecule-basis combinations. Here $n=2N$, each sector rank is at most $N-1$, and the combined ceiling is $2(N-1)$. Term counts exclude the identity. Every term satisfies both parity restrictions. The column $r_X(S)$ gives the rank of the commuting family found by the searches, and $r_X(H)$ the rank of the $X$ block of the whole Hamiltonian support. Since every subset has rank at most $r_X(H)$, and $r_X(S)=r_X(H)$ in every row, each family is maximal and the shortfall $2(N-1)-r_X(H)$ is exact. The orbitals are canonical RHF orbitals computed without point-group symmetry, so degenerate orbitals are mixed (see text). NH$_3$ is omitted for the reason given in Sec.~\ref{sec:results}.}
\label{tab:two-sector}
\begin{tabular}{@{}llrrrrrrrr@{}}\toprule
Molecule & Basis & $n$ & $N$ & Terms & $r_X^\alpha$ & $r_X^\beta$
  & $r_X(S)$ & $r_X(H)$ & $2(N{-}1)$ \\\midrule
H$_2$    & STO-3G & 4  & 2  & 14     & 1 & 1 & 1 & 1 & 2  \\
LiH      & STO-3G & 12 & 6  & 630    & 5 & 5 & 8 & 8 & 10 \\
BeH$_2$  & STO-3G & 14 & 7  & 665    & 6 & 6 & 9 & 9 & 12 \\
H$_2$O   & STO-3G & 14 & 7  & 1{,}085  & 6 & 6 & 10 & 10 & 12 \\
HF       & STO-3G & 12 & 6  & 630    & 5 & 5 & 8 & 8 & 10 \\
N$_2$    & STO-3G & 20 & 10 & 2{,}950  & 9 & 9 & 16 & 16 & 18 \\
CO       & STO-3G & 20 & 10 & 5{,}850  & 9 & 9 & 17 & 17 & 18 \\
CH$_4$   & STO-3G & 18 & 9  & 6{,}891  & 8 & 8 & 16 & 16 & 16 \\\midrule
LiH      & 6-31G  & 22 & 11 & 8{,}757  & 10 & 10 & 19 & 19 & 20 \\
BeH$_2$  & 6-31G  & 26 & 13 & 9{,}203  & 12 & 12 & 22 & 22 & 24 \\
H$_2$O   & 6-31G  & 26 & 13 & 12{,}731 & 12 & 12 & 22 & 22 & 24 \\
HF       & 6-31G  & 22 & 11 & 8{,}757  & 10 & 10 & 19 & 19 & 20 \\\bottomrule
\end{tabular}
\end{table*}

We do not count the identity in these figures. NH$_3$ is left out for the reasons given in Sec.~\ref{sec:results}. In cases where $N\geq6$ the families hit both sector projection ceilings. CH$_4$ also meets the combined one. The 6-31G examples show the same parity pattern over larger supports.

The shortfall $2(N-1)-r_X(H)$ has a symmetry origin. The binary vectors orthogonal to every $X$ mask of the support are spanned by the two sector-parity vectors plus one for each independent $\mathbb Z_2$ character of the point group the orbitals maintain. Each absent binary direction is the conserved occupation parity of an irreducible representation. Take LiH: the $\pi_x$ and $\pi_y$ parities leave $10-2=8$. For H$_2$O in $C_{2v}$ it is $12-2=10$ from the $B_1$ and $B_2$ parities, and for BeH$_2$ in $D_{\infty h}$ the $\pi_x$, $\pi_y$ and $g/u$ parities give $12-3=9$. Mixing degenerate orbitals costs some of this parity and inflates the term counts. Using symmetry-adapted orbitals, the N$_2$, CO and CH$_4$ operators come in at 2238, 4426 and 2211 terms respectively with $r_X(H)$ of 15, 16 and 14. The 6-31G LiH, HF and BeH$_2$ have 6869, 6869 and 7611 terms ($r_X(H)=18$, 18, 21). H$_2$O and the nondegenerate STO-3G are unaltered.

\section{Searching for high X-rank commuting families}
\label{sec:search}
Our approach builds individual readable settings. It is not concerned with minimizing the cover. We rely on the discrete extension of the artificial bee colony and genetic algorithm (ABC+GA) architecture found in RANGE~\cite{RANGE} and the companion certificate work~\cite{Zahariev2026CCF}. A candidate is made up of Hamiltonian terms. Through crossovers and moves we ensure pairwise commutation, then a greedy step is applied to find any additional independent $X$ directions. The fitness measure is the exact binary rank of the $X$ block, making every candidate we evaluate a bona fide commuting family.

For the comparisons here we have run three separate searches: discrete RANGE, a genetic algorithm and a multistart greedy method. They all come up with the same best rank for any given operator and we have verified the commutation of each reported family. This concordance speaks to the reproducibility of the numbers. Their maximality is guaranteed by the fact that no subset can go beyond $r_X(H)$ of the whole Hamiltonian support, a level reached by every entry in Table~\ref{tab:two-sector}. An independent multistart greedy search reproduces every value. A separate optimization problem is the subject of the continuous rotation-angle search in Sec.~\ref{sec:separation}.

The small-molecule data are collected in Table~\ref{tab:xrank_saturation}. Take CH$_4$ for instance: in the symmetry-broken orbital gauge of Sec.~\ref{sec:parity}, RANGE will hit $r_X=16$ (that is, $2(N-1)$) in a single generation or 48 evaluations, and it does so reproducibly across random seeds. The symmetry-adapted operator for CH$_4$ has a whole-support rank of 14, while the heuristic's sorted-insertion groups only go as high as 8. A more focused search uncovers what that cover misses and meets the parity ceiling set by Theorem~\ref{thm:parity}. Any other small-molecule deficits are no more than three. They are exact in that each family makes it to $r_X(H)$.
\begin{table}[!t]\centering\footnotesize\setlength{\tabcolsep}{3pt}
\caption{Jordan-Wigner rank-search comparisons in the STO-3G basis. The greedy column is the largest $X$-rank of a group in the sorted-insertion cover, which inserts the terms in order of decreasing coefficient magnitude into the first group with which they fully commute. RANGE, a separate memetic genetic algorithm, and multistart seeded greedy search return the same best ranks, and each equals the whole-support rank $r_X(H)$, so the deficit $2(N-1)-r_X(H)$ is exact. CH$_4$ alone attains the ceiling in this set, in the orbital gauge stated in Sec.~\ref{sec:parity}. NH$_3$ is omitted for the reason given in Sec.~\ref{sec:results}.}
\label{tab:xrank_saturation}
\begin{tabular}{@{}llccccc@{}}\toprule
Molecule & Symmetry & Qubits & $2(N{-}1)$ & Greedy & $r_X^{\rm best}$ & Deficit \\\midrule
LiH & $C_{\infty v}$ & 12 & 10 & 5 & 8 & 2 \\
BeH$_2$ & $D_{\infty h}$ & 14 & 12 & 6 & 9 & 3 \\
H$_2$O & $C_{2v}$ & 14 & 12 & 6 & 10 & 2 \\
CH$_4$ & $T_d$ & 18 & 16 & 8 & 16 & 0 \\
N$_2$ & $D_{\infty h}$ & 20 & 18 & 9 & 16 & 2 \\\bottomrule
\end{tabular}
\end{table}

\begin{observation}[Sector-projection saturation in the tested set]
\label{obs:saturation}
We have looked at nine operators in this section and Sec.~\ref{sec:felements}, from 12 to 36 qubits. In every case the top family has $r_X^\alpha=r_X^\beta=N-1$. One can see the shortfall from $2(N-1)$ as $r_X^\alpha+r_X^\beta-r_X$, which is to say the dependence between the binary sector projections. It is not a matter of unsaturated projection rank in either sector. Since all families come to the rank $r_X(H)$ of their support, the nine values are maximal and the deficits exact. For the small molecules the reason is the point-group parities of Sec.~\ref{sec:parity}. For the heavy-element operators at the block-invariant symmetry shift (BLISS)~\cite{Loaiza2023} stage the dependence counts are 1, 1, 2 and 2, and no symmetry analysis was made.
\end{observation}

It is important to distinguish three things: the algebraic ceiling, the rank of the whole support a dedicated search can get to, and what a practical cover exposes. The CH$_4$ projections are independent, $8+8$ for 16. In the small molecules the deficit is the number of independent $\mathbb Z_2$ symmetries the orbital gauge holds on to, per Sec.~\ref{sec:parity}.

\section{$X$-support structure in heavy-element Hamiltonians}
\label{sec:felements}
Our heavy-element operators are CeO ($4f^15d^1$), NdO ($4f^3$), UO$_2^{+}$ ($5f^1$) and UO$_2^{2+}$ ($5f^0$). They come from a complete active space self-consistent field (CASSCF) workflow, state-averaged with frozen-natural-orbital truncation and symmetry-shift norm reduction. The companion certificate study~\cite{Zahariev2026CCF} has the full report on the coefficient audit, sampling costs and electronic-structure provenance. We have the audited operators at five points: the linear particle-number shift, the combined linear $N$, $N^2$, $S_z^2$ and $NS_z$ shifts, full BLISS, BLISS plus the extra symmetry shifts, and the $\mathbb Z_2$-tapered forms. Ranks and counts here are for nonidentity terms with coefficient magnitude above $10^{-10}$.

There is no spin-orbit coupling in these scalar relativistic, spin-free Hamiltonians, so the two spin populations are conserved in keeping with Theorem~\ref{thm:parity}. We compare with $2(N-1)$ using the untapered Jordan-Wigner in interleaved mode order (30 to 36 qubits). An even $X$ weight was verified for every Pauli term in each sector. Tapered operators are left out of such a comparison because tapering alters the Clifford frame and the sector parities do not hold.
\begin{table*}[!t]\centering\footnotesize\setlength{\tabcolsep}{5pt}
\caption{$X$-rank structure of the audited untapered heavy-element operators after the full block-invariant symmetry shift, with nonidentity terms of coefficient magnitude above $10^{-10}$. Greedy denotes the maximum rank in the sorted-insertion cover, $r_X(S)$ the rank of the best commuting family found by the searches, and $r_X(H)$ the rank of the $X$ block of the whole operator support. Since $r_X(S)=r_X(H)$ in every row, each family is maximal for its operator and cutoff. The parenthesized dependence is $r_X^\alpha+r_X^\beta-r_X(H)$. The NdO rank depends on the coefficient cutoff (see text).}
\label{tab:felement_xrank}
\begin{tabular}{@{}lccccccccc@{}}\toprule
System & Configuration & $N$ & Qubits & Terms & Greedy & $r_X(S)$ & $r_X(H)$ & $2(N{-}1)$ & Sector ranks (dependence) \\\midrule
CeO & $4f^15d^1$ & 15 & 30 & 33{,}843 & 14 & 27 & 27 & 28 & $14{+}14$ (1) \\
NdO & $4f^3$ & 16 & 32 & 23{,}788 & 15 & 29 & 29 & 30 & $15{+}15$ (1) \\
UO$_2^{+}$ & $5f^1$ & 17 & 34 & 17{,}429 & 15 & 30 & 30 & 32 & $16{+}16$ (2) \\
UO$_2^{2+}$ & $5f^0$ & 18 & 36 & 31{,}584 & 17 & 32 & 32 & 34 & $17{+}17$ (2) \\\bottomrule
\end{tabular}
\end{table*}

Table~\ref{tab:felement_xrank} shows the untapered full-BLISS operators. The sorted-insertion covers run to 391, 323, 231 and 360 groups with mean ranks of 10.95, 10.58, 11.16 and 12.68, and maximums of 14, 15, 15 and 17 against the one-sector ceilings of $N-1$. This is the same pattern one sees with CH$_4$ for these covers, though not a rule for sorted insertion in general.

A dedicated search yields $r_X$ of 27/28 for CeO, 29/30 for NdO, 30/32 for UO$_2^{+}$ and 32/34 for UO$_2^{2+}$. These are the rank $r_X(H)$ of the whole support, making each family maximal for its operator and cutoff. All four show saturated projections, $14+14$ through $17+17$, with the aforementioned dependence counts. This is in addition to the five small-molecule operators of Sec.~\ref{sec:search}. Together they give the nine cases of projection saturation in Observation~\ref{obs:saturation}. With NdO the rank is a function of the coefficient cutoff. An operator with no cutoff, or one that has been linearly shifted at $10^{-10}$, will show $r_X(H)=30$ (or $2(N-1)$). The full-BLISS operator is another matter. At that same cutoff its rank is 29 since BLISS has put 18 of the weaker X masks, with coefficients in the $10^{-10}$ to $10^{-8}$ range, below it. Should the cutoff be $10^{-8}$ then NdO comes to 28 and UO$_2^{+}$ to 29 at all stages. One could say NdO only reaches the ceiling on account of the $10^{-10}$ to $10^{-8}$ couplings left over from the frozen-natural-orbital truncation.

As for symmetry shifts, let us look at the diagonal shift first:
\[
H'=H-\lambda_N(\widehat N-N_t)-\lambda_{S_z}(\widehat S_z-S_{z,t}),
\]
where $\widehat N$ and $\widehat S_z$ stand for the particle-number and spin-projection operators, with $N_t$ and $S_{z,t}$ as their target-sector values. The only coefficients to change are identity and single-$Z$. Because there is no $X$ support in these strings, the $X$ masks of the off-diagonal terms are not affected, nor are the ranks of any fixed subsets thereof. Removing or adding zero-$X$ rows makes no difference to the rank. A review of the operators bears this out for both the combined quadratic and linear shifts, which have identical off-diagonal term sets and coefficients. The BLISS stage is an exception. Its one-body component multiplies the deviation in particle number into the one-body operators and appends $Z$-decorated off-diagonal strings to what is already on the $X$ masks. At the $10^{-10}$ cutoff, for instance, CeO puts on 532 off-diagonal strings and gives up 36, with coefficient changes as high as 0.585 hartree. NdO and UO$_2^{+}$ lose 112 and 24 respectively, while UO$_2^{2+}$ is net positive with 272 gained and 76 lost. Still, if one looks at the set of $X$ masks without a cutoff, it is the same for every system across the four stages (6459, 4598, 3440 and 6297 masks). Thus $r_X(H)$ of 27, 30, 30 and 32 is invariant to the stage for each operator, and the only reason for the reported figures to differ is the aforementioned sensitivity to the cutoff in NdO and UO$_2^{+}$. That said, an unchanging rank does not mean a regenerated greedy cover will have the same statistics.
\begin{table}[!t]\centering\footnotesize\setlength{\tabcolsep}{3pt}
\caption{CeO rank statistics for the audited operators at three stages, with nonidentity terms of coefficient magnitude above $10^{-10}$. The mean is over the groups of the sorted-insertion cover. Tapering changes the Clifford frame, removes three qubits, and leaves the term count and the commutation graph unchanged. In the tapered frame the sector parities no longer hold, the projections are $14+13$, and the rank is not compared with the untapered ceiling.}
\label{tab:xrank_stages}
\begin{tabular}{@{}lcccccc@{}}\toprule
Stage & Qubits & Terms & Groups & Greedy & Mean & $r_X(H)$ (dependence) \\\midrule
Linear $N$ shift & 30 & 33{,}361 & 379 & 14 & 10.92 & 27 (1) \\
Full BLISS & 30 & 33{,}843 & 391 & 14 & 10.95 & 27 (1) \\
Full BLISS, tapered & 27 & 33{,}843 & 391 & 14 & 10.95 & 27 \\\bottomrule
\end{tabular}
\end{table}

Table~\ref{tab:xrank_stages} presents the CeO statistics for three of the stages. Tapering will alter the Clifford frame and take away three qubits but the term count and commutation graph are left as they are, so the cover statistics remain those of the untapered operator. The sector projections now read $14+13$ and the parity ceiling is no longer in force. In the earlier form of this table the stage labels were for the linearly and quadratically shifted operators whose 33,361 and 33,359 terms are matched precisely by the audited operators at $10^{-10}$.

These findings on $X$ support are distinct from the correlation-rank obstruction put forward in Sec.~\ref{sec:separation}. As for whether there is a limitation to the orbital settings of these operators, that would call for an explicit analysis of the fragments that Gaussian rotations can read. Neither the rank searches here nor the companion work comparing fully and qubit-wise commuting collections have shown such a limitation to exist.

\section{Constructed covers and the matter of sampling}
\label{sec:results}
\subsection{Setting counts for the molecular examples}
A comparison of the constructed CSA orbital covers against the graph-based Clifford covers employed in this work is set out in Table~\ref{tab:circuit-counts}. One will note that the Clifford heuristic requires fewer settings for H$_2$, LiH and BeH$_2$, though more are needed for H$_2$O. Such results are only indicative of the particular covers being compared. A larger heuristic cover does not mean the Clifford class as a whole is any less capable. As for the Clifford counts $K_{\CC}$ for LiH, BeH$_2$ and H$_2$O, these are the constructed counts bounded from below in Table~\ref{tab:theta}. An 18-setting LiH cover was obtained through tabu search on the anticommutation graph. Nine of its groups have $X$ or $Y$ support running the full length of the twelve qubits, precluding their confinement to a single six-qubit spin sector, and four reach the whole-support rank $r_X(H)=8$ of Table~\ref{tab:two-sector}.
\begin{table}[!t]\centering\footnotesize\setlength{\tabcolsep}{4pt}
\caption{Setting counts obtained with the CSA and Clifford grouping implementations. Here $n$ is the qubit count, $N$ the number of spatial orbitals, and $\Delta=K_{\CSA}-K_{\CC}$, so that a positive $\Delta$ favors the constructed Clifford cover. Term counts include the identity. The counts are specific to the two heuristics and are upper bounds on the class minima. The CSA counts are the published values~\cite{Yen2023}, reproduced by the fragment-count script of the data release at the published geometries. NH$_3$ is omitted (Sec.~\ref{sec:results}).}
\label{tab:circuit-counts}
\begin{tabular}{@{}lrrrrrr@{}}\toprule
Molecule & $n$ & $N$ & Terms & $K_{\CSA}$ & $K_{\CC}$ & $\Delta$ \\\midrule
H$_2$   & 4  & 2 & 15   & 4  & 2   & $+2$   \\
LiH     & 12 & 6 & 631  & 22 & 18  & $+4$   \\
BeH$_2$ & 14 & 7 & 666  & 29 & 20  & $+9$   \\
H$_2$O  & 14 & 7 & 1086 & 29 & 34  & $-5$   \\\bottomrule
\end{tabular}
\end{table}

The term counts in this table include the identity, whereas Table~\ref{tab:two-sector} excludes it, which accounts for their one-term difference. We have left NH$_3$ out of the quantitative parity and cover comparisons since its Pauli-term count depends on the orbital chosen in the degenerate $E$ subspace of $C_{3v}$ (see Appendix~\ref{sec:integer}). The CSA figures are those published in~\cite{Yen2023} and can be reproduced with our fragment-count script at the stated geometries. All parity ranks and Clifford cover counts are taken from our operator builds. Nothing is inferred across them.

\subsection{On the cost of fully commuting settings}
The companion certificate study~\cite{Zahariev2026CCF} puts a set of qubit-wise commuting (QWC) settings up against the union of those and some fully commuting (FC) candidates. In a QWC group, any two strings will share an identity factor or the same Pauli factor on each qubit. The companion optimization of Eq.~\eqref{eq:certifiedcost} applies the same covariance rules and term caps per setting. The leading shot saving is
\[
1-\left(\frac{\Phi_{\rm QWC+FC}}{\Phi_{\rm QWC}}\right)^2.
\]
Take the 18-qubit CH$_4$ benchmark: the measurement constant drops from $10.795$ to $8.455$, which is a 39\% shot saving. For the four heavy-element operators we refer to the audited values in~\cite{Zahariev2026CCF}. Earlier numbers were put aside after a coefficient audit and are not repeated. These are numerical results shared with the companion study, not a second independent study.

While entangling settings are known to aid estimation efficiency~\cite{HamamuraImamichi2020}, what we report is the effect for the given collections at fixed precision under the declared Hartree-Fock covariance model. To make that claim for every stabilizer setting would demand a lower bound beyond the finite pools.

These numbers and the two-orbital theorems are complementary but address different issues. The savings pertain to a comparison of finite QWC and QWC+FC under one covariance model and do not point to correlation rank as the cause. The theorems, for their part, demonstrate that no single-setting substitute exists in the restricted orbital class for a given Clifford-readable observable.

\section{Discussion}
Correlation rank will only give the precise orbital setting count if the local terms are aligned to the same measurement axes. The singular value decomposition provides the axes for reconstructing the correlation matrix, and any one-qubit components thereon can be handled in the same settings. Any two residual local components are accommodated in one more, bringing the exact reconstruction for the two-orbital model to a close.

Sampling is better with joint readout because it holds onto the correlations between measured terms. With maximal mixing there is zero covariance for distinct Pauli scores and the orbital measurement constant is simply the sum of singular values. But in other states a joint Clifford measurement keeps covariances that would be lost in separate batches. In the case of an equal mixture of logical triplets, the exchange observable is made deterministic by these covariances despite individual fluctuations in the orbital fragments. This cancellation is what allows the shot advantage to vary even as the orbital reconstruction count stays at three.

Then there is the question of exchange separation, driven by the need to preserve the spin populations separately. Spin-sector rotations on their own limit the orbital axes one can access. Mix the spins and an occupation measurement can read isotropic exchange directly. An orbital readout of a continuously rotated single-axis observable may require fewer settings. A good design will combine orbital and Clifford settings to take advantage of both.

As for atomic decompositions, consider the space of fragments readable in setting $k$. Let $\mathcal B_k$ be the set of fragments $F$ where $\sqrt{c_k\Var_\rho(F)}\leq1$, treating null directions as in~\cite{Zahariev2026CCF}. The gauge generated by these normalized balls is the measurement constant $\Phi$, minimizing the weight of a reconstruction from the readable fragments. The number of active settings is then the nonconvex group-support measure. This puts the pair in line with the cardinality or rank versus convex gauge discussion in~\cite{Chandrasekaran2012} without equating an arbitrary measurement gauge to a nuclear norm. That identification is exact in Proposition~\ref{prop:price} since the target is the logical correlation block alone, the state is maximally mixed and costs are unity, and the full restricted orbital family is available. Under those conditions the rank-constrained problem has the complete Eckart-Young error curve. It is a general rule that one cannot improve the optimum of a collection by adding to it, provided the covariances and costs of what is already in place are fixed.

Statistical gains are always a matter of the model. The molecular comparisons put forward in Sec.~\ref{sec:results} are no exception: for any given pair, the same operator, covariance prescription and per-shot cost convention are in play. What is reported are numerical certificates, not inferences from an isolated basis-change cost or some fragment-count asymptotic. They do not represent an optimum across all stabilizer settings or a class of molecular orbital rotations. Take the 18-qubit CH$_4$ benchmark for instance: if fully commuting settings are allowed, 39\% of the leading shots can be saved. In the case of H$_4$, the complete-family certificate from the companion study is the molecular equivalent of a fixed-sector separation. A decomposition is made from the union of the ten (or nine at a stretched geometry) fully commuting settings found by sorted insertion and the 6561 Pauli-axis product settings of the complete family. This uses at least 61\% fewer leading shots than any product-only estimator would at either tested geometry, subject to the stated target-term score and covariance model.

Then there is the question of statistical efficiency versus the price of execution. A Clifford readout brings with it no measurement-specific $T$ gates (Corollary~\ref{cor:zero-T}), but entangling gates, calibration, routing and errors are still relevant. Equation~\eqref{eq:certifiedcost} allows for the execution time of a setting via the per-shot costs $c_k$, though it makes no provision for gate errors. One might therefore find a hardware-dependent model putting a higher per-shot cost on a Clifford setting even with its zero $T$ count.

Other costs are inherent to transition properties. As Appendix~\ref{app:controlled} demonstrates, a controlled-Pauli insertion is a Clifford operation and so contributes no $T$ gates to a Hadamard test. But the ancilla protocol, the need for repetitions, controlled state preparation and the estimation of different matrix elements still require resources. Any statement on total cost has to apply weighted accounting to the whole protocol, not just the Pauli insertion.

One must also distinguish between search quality and class limitations. A heuristic cover will give an upper bound on the setting count while a numerical theta witness provides a lower bound as per its feasibility rules. The gap is quantified in the molecular comparisons of Appendix~\ref{sec:integer}. The two-sector parity ceiling and one-sector dimension bound of Sec.~\ref{sec:parity} put constraints on $X$ support without barring orbital accessibility, which is only ruled out by the correlation-rank argument. None of this is proof of efficient near-optimal grouping for larger, arbitrary graphs.

The extensions that follow are thus of a specific nature. They call for more structural or numerical evidence than a wide interpretation of the present theorems can offer: putting a price on molecular unions of Clifford and orbital settings~\cite{Yen2023}, characterizing spin-mixing and number-nonconserving Gaussian settings, strengthening support diagnostics through the incidence structure of orbital pairs, getting better integer lower bounds for large graphs, or comparing the estimators induced by classical-shadow schemes within the same declared model~\cite{Huang2020}. And to extend the exact count into larger orbital spaces takes more than a rank calculation. A rotated occupation basis gives a constrained set of correlation directions, meaning their rank may well be an undercount of the settings required. To characterize the accessible directions, one would have to go beyond the reconstruction analysis for two orbitals per spin, and comparing optimal sampling costs for molecular observables would additionally require their covariances.

\section{Conclusions}
Consider a general observable in the fixed $(1,1)$ sector of two spatial orbitals per spin. The minimum number of orbital settings that preserve each spin population is precisely the larger rank of its two augmented correlation matrices, a figure reached by a constructive singular value decomposition that includes the one-qubit terms. The tails of the singular values then set a stable lower bound on the error of the full observable and determine the correlation approximation error exactly.

An observable that is Bell-diagonal with three nonzero couplings calls for three settings in that orbital class, but an explicit Clifford readout on four physical qubits will measure it. That is a separate issue from statistical cost. With unit per-shot costs for isotropic exchange in the logical states of the isotropic family, the ratio of leading cost comes to $3(1+w)/(1+3w)$ for $-1/3<w<1$. The maximally mixed case is where the factor of three applies. Mix the spin sectors and that particular exchange separation is lost.

The rest of the analyses serve more limited purposes. Spin parity puts a $2(N-1)$ cap on the Pauli $X$ support of a commuting family, a ceiling the discrete search for CH$_4$ in a symmetry-broken orbital gauge can reach. The rank of the Hamiltonian support as a whole is enough to certify the maximality of every molecular search value, whether for small molecules or audited heavy-element operators. The semidefinite witnesses put a lower bound on molecular Clifford setting counts. When it comes to the full nonidentity Pauli set, the theta bound, the integer cover and the counting bound are in accord at $2^n+1$. As for the molecular comparisons cited here, they pertain to fully commuting additions to qubit-wise commuting sets, not an optimum over an orbital class. The companion certificate study has the bounds on sampling costs for finite collections. And while primitives like controlled-Pauli insertions and Clifford basis changes have a zero measurement-specific $T$ count, any estimate of resources has to make room for the other execution costs, from repeated sampling to state preparation.

\begin{acknowledgments}
This work was supported by the U.S. Department of Energy, Office of Science, through the Ames National Laboratory under Contract No. DE-AC02-07CH11358. Computational resources were also provided by Iowa State University. The authors acknowledge partial support by ORNL's LDRD and VSO programs. This manuscript has been authored by UT-Battelle, LLC, under contract DE-AC05-00OR22725 with the U.S. Department of Energy (DOE). The US government retains and the publisher, by accepting the article for publication, acknowledges that the US government retains a nonexclusive, paid-up, irrevocable, worldwide license to publish or reproduce the published form of this manuscript, or allow others to do so, for US government purposes. DOE will provide public access to these results of federally sponsored research in accordance with the DOE Public Access Plan (https://www.energy.gov/doe-public-access-plan). This research used resources (Director's discretionary allocation CHM238) of the Oak Ridge Leadership Computing Facility at the Oak Ridge National Laboratory, which is supported by the Office of Science of the U.S. Department of Energy under Contract No. DE-AC05-00OR22725. This research used resources of the National Energy Research Scientific Computing Center (NERSC), a Department of Energy Office of Science User Facility, supported by Contract No. DE-AC02-05CH11231 using NERSC award m4621 (ERCAP0036406). The authors thank Ulrike Stege for discussions on graph coloring formulations of Pauli measurement grouping and Artur Izmaylov for foundational work on the CSA framework.
\end{acknowledgments}

\section*{Author contributions}
F.Z. conceived the study, carried out the derivations, and developed the verification scripts. V.-A.G. contributed to the problem formulation and the interpretation of the results. F.Z. drafted the manuscript. Both authors revised the text and approved the final version.

\appendix
\section{Bounds on molecular Clifford setting counts}\label{sec:integer}
Counting Clifford settings can be formulated as a graph-covering problem. For $H=\sum_i c_iP_i$, the commutativity graph $G_C(H)$ has one vertex for each nonzero nonidentity Pauli term, with an edge between commuting strings. A jointly readable group forms a clique, meaning that every pair of its vertices is connected. A clique cover of this graph is equivalent to a coloring of its complement. Any explicit grouping thus provides an upper bound on the minimum setting count~\cite{Verteletskyi2020,Gokhale2020,Jena2022}. Degree-of-saturation ordering (DSATUR)~\cite{Brelaz1979} provides an inexpensive coloring baseline, and TabuCol~\cite{Hertz1987} supplies a local-search alternative. A valid improved cover reduces the number of Clifford basis-change circuits, and the measurement-specific $T$ count remains zero.

In the binary representation of Appendix~\ref{app:symplectic}, a Pauli string up to phase is a vector $(x\mid z)$ in $\F_2^{2n}$, where $\F_2$ is the binary field. Two strings commute exactly when the symplectic form $\omega(u,v)=x_u\cdot z_v+z_u\cdot x_v \pmod 2$ of their vectors vanishes, so a commuting family spans an isotropic subspace, on which $\omega$ vanishes identically.

\begin{proposition}[Clifford conjugation preserves the grouping problem]
\label{prop:clifford-graph}
Suppose $H_1=\sum_i c_iP_i$ and $H_2=\sum_i c_iP_i'$, with $P_i'=CP_iC^\dagger$ for a Clifford unitary $C$. Their commutativity graphs are isomorphic as coefficient-weighted graphs. With the signs included in the conjugated representatives, the coefficient multiset is unchanged.
\end{proposition}
\begin{proof}
Clifford conjugation maps Pauli strings one to one and preserves their commutators:
\begin{equation}
[P_i,P_j]=0\quad\Longleftrightarrow\quad
[CP_iC^\dagger,CP_jC^\dagger]=0.
\label{eq:commutation-invariance}
\end{equation}
The representation of $H_2$ assigns coefficient $c_i$ to the image of $P_i$, so the bijection also preserves the vertex weights.
\end{proof}
When Pauli strings are instead represented by canonical positive tensor products of $I$, $X$, $Y$, and $Z$, the signs of the conjugates move into their coefficients, and the coefficient magnitudes are unchanged. Either convention carries a clique cover to a cover of the same size. The proposition establishes this combinatorial invariance and does not concern the invariance of a sampling model with independently specified state covariances.

\begin{remark}[The ambient symplectic geometry]
\label{rem:symplectic}
The nonidentity Pauli strings modulo phase correspond to the nonzero vectors of $\F_2^{2n}$, and commutation joins symplectically orthogonal vectors. A clique spans an isotropic subspace, whose dimension is at most $n$. Accordingly, a setting contains at most $n$ independent generators, although it can read many of their products. A molecular operator selects an induced subgraph from this ambient geometry, and for spin-conserving operators Sec.~\ref{sec:parity} further restricts the $X$ support within each spin sector. This geometric description organizes the family of settings, and the later arguments require no additional finite-geometry machinery.
\end{remark}

The invariance applies to every pair of Pauli lists related by a Clifford conjugation, including qubit permutations. For the benchmark encodings and orderings tested here, the edge counts, degree sequences, and sorted coefficient magnitudes agree between the Jordan-Wigner and Bravyi-Kitaev lists. We report that agreement as an implementation check and not as a statement about unrestricted changes of encoding.

The Lov\'asz theta calculation supplies a lower bound without reference to the chosen cover. It satisfies $\vartheta(G_C)\leq\chi(\overline{G_C})$, where $\chi$ is the minimum number of colors~\cite{Lovasz1979}. A feasible theta matrix is positive semidefinite, has unit trace, and has zero entries on commuting edges. Its objective is no greater than the theta optimum and therefore no greater than the minimum setting count.

Before evaluating a numerical lower bound, we repair the matrix constraints. The checker restores the required edge zeros, shifts the matrix to make it positive semidefinite, and rescales it to unit trace. It then verifies the input identity, edge entries, trace, and eigenvalues within the specified tolerance. These checks use floating-point arithmetic rather than interval arithmetic.

\begin{table}[!t]\centering\footnotesize\setlength{\tabcolsep}{2.2pt}
\caption{Bounds on the Clifford setting counts for three stored molecular Pauli expansions. Heuristic groupings give the constructed upper counts. Positive-semidefinite theta optimization gives the numerical lower values, whose ceilings yield integer lower bounds at the specified tolerance. The last column is the constructed count divided by the integer lower bound. The LiH, BeH$_2$, and H$_2$O operators have 630, 665, and 1085 nonidentity terms on 12, 14, and 14 qubits, respectively. These setting-count bounds do not require a state covariance model.}\label{tab:theta}
\begin{tabular}{@{}lrrrr@{}}\toprule
System&Constructed&Theta bound&Integer bound&Ratio\\\midrule
LiH&18&15.9999&16&1.13\\
BeH$_2$&20&14.664&15&1.33\\
H$_2$O&34&23.349&24&1.42\\\bottomrule
\end{tabular}\end{table}

These bounds apply to the signed operators identified in Table~\ref{tab:theta}. Earlier NH$_3$ groupings are excluded because the orbital choices within a degenerate subspace differ, so no numerical improvement is attributed to a particular algorithm on the basis of that sequence (Sec.~\ref{sec:results}). By Proposition~\ref{prop:clifford-graph}, the bounds are unchanged under Clifford relabeling of the operators. Comparing sampling costs would also require the state covariance and the cost of executing each setting.

Greedy searches find pairwise anticommuting families of sizes 7--9 in these graphs, and each such family also gives a combinatorial lower bound on the setting count, since no two of its members can share a setting. The theta witnesses are stronger. They place the constructed covers within factors 1.13--1.42 of their integer lower bounds, and at the reported numerical resolution LiH lies no more than two settings above the optimum. These graphs have roughly one thousand vertices and are tractable on one node. Larger graphs will require a more economical semidefinite representation, for example one that uses sparsity or low rank. High $X$-rank (Sec.~\ref{sec:parity}) can identify settings for further study, but it is neither a criterion for orbital accessibility nor a substitute for a covariance-aware sampling objective.

\subsection{The complete nonidentity Pauli set}
For the full Pauli graph, symmetry makes an exact result possible.
\begin{proposition}[Exact complete-Pauli cover]
\label{prop:complete}
Let $G_n$ be the commutativity graph of the $4^n-1$ nonidentity $n$-qubit Pauli strings modulo phase, and let $\overline\chi(G_n)$ be its minimum clique-cover number. Then
\begin{equation}
\vartheta(G_n)=\overline\chi(G_n)=2^n+1.
\label{eq:fullpauli}
\end{equation}
A partition into $2^n+1$ maximal commuting classes attains the minimum.
\end{proposition}
\begin{proof}
For $n=1$ the three Pauli strings pairwise anticommute, so the graph has no edges and the claim is immediate. For $n\geq2$, the complement graph joins nonzero symplectic vectors with $\omega(u,v)=1$. It is strongly regular with parameters $(4^n-1,2^{2n-1},2^{2n-2},2^{2n-2})$ and eigenvalues $2^{2n-1}$ and $\pm2^{n-1}$~\cite{BrouwerHaemers2012}. Since the adjacency matrices satisfy $A(G_n)=J-I-A(\overline G_n)$, with $J$ the all-ones matrix, the commutativity graph has degree $d=2^{2n-1}-2$, second eigenvalue $2^{n-1}-1$, and least eigenvalue $s=-(2^{n-1}+1)$.

Witt's theorem gives vertex and edge transitivity under the symplectic group $\operatorname{Sp}(2n,2)$, which makes the eigenvalue bound on the theta function of a regular graph tight~\cite{Lovasz1979}. With $N_v=4^n-1$ vertices,
\[
\vartheta(G_n)=\frac{N_v(-s)}{d-s}
=\frac{(4^n-1)(2^{n-1}+1)}{2^{2n-1}+2^{n-1}-1}
=2^n+1.
\]
To construct a matching cover, identify $\F_2^{2n}$ with pairs $(x,z)$ of elements of the field $\F_{2^n}$ with $2^n$ elements and use the equivalent nondegenerate form
\[
\omega((x,z),(x',z'))=\operatorname{Tr}(xz')+\operatorname{Tr}(zx'),
\]
where the trace maps the extension field to $\F_2$. The sets $L_a$ of pairs $(x,ax)$ with $x\ne0$, one for each $a$ in $\F_{2^n}$, together with the set $L_\infty$ of pairs $(0,z)$ with $z\ne0$, partition the nonzero vectors. Each set is isotropic because $\omega((x,ax),(x',ax'))=2\operatorname{Tr}(axx')=0$ in characteristic two. This symplectic spread is the known commuting-class partition associated with mutually unbiased bases~\cite{Bandyopadhyay2002} and supplies the attaining cover.
\end{proof}

Elementary counting gives the same integer optimum: a maximal commuting class contains $2^n-1$ nonidentity strings, so at least $(4^n-1)/(2^n-1)=2^n+1$ classes are needed. The point of the proposition is the coincidence of the counting bound, the unweighted fractional clique-cover number, the theta function, and the integer clique-cover number. This fractional cover is distinct from the covariance-weighted objective of Eq.~\eqref{eq:certifiedcost}.

A molecular induced subgraph need not retain that coincidence. Its counting lower bound requires a certified upper bound on the largest clique, and finding a large commuting family is insufficient: the 78-term LiH family found by the searches does not certify a nine-setting floor. The theta witness establishes a sixteen-setting lower bound for LiH without assuming a maximum clique size.

The ambient symmetry also suggests a route to larger calculations. Its rank-three group action reduces the semidefinite program to a three-dimensional commutant. Molecular term sets invariant under a smaller subgroup may admit related block reductions~\cite{GatermannParrilo2004,Schrijver2005}. The graph parameters, the spread construction, and the theta equality are checked exactly for $n=2,3,4$ in the supporting implementation, and the minimum cover of $2^n+1$ classes and the largest pairwise anticommuting set of size $2n+1$ are verified by exhaustive computation for $n=2$ and $3$. Applying symmetry reduction to larger molecular graphs remains a computational extension rather than a result of the present study.
\section{Clifford implementation of a controlled Pauli}
\label{app:controlled}
A controlled Pauli can be implemented entirely with Clifford operations. The following identity accounts for the controlled insertion in an ancilla measurement, where an auxiliary qubit controls the Pauli operation. The full protocol must also account for state preparation and repeated measurements.
\begin{lemma}[Controlled-Pauli decomposition]
For $P=\bigotimes_{j=1}^nP_j$, with $P_j$ chosen from $I$, $X$, $Y$, and $Z$, the operation
\begin{equation}
C\text{-}P=|0\rangle\langle0|_a\otimes I^{\otimes n}
+|1\rangle\langle1|_a\otimes P
\end{equation}
is Clifford on $n+1$ qubits. If $P$ contains $k$ nonidentity factors, an implementation uses no more than $k$ controlled-$Z$ (CZ) gates and $2k$ single-qubit Clifford operations.
\end{lemma}
\begin{proof}
For each nonidentity factor, select $C_j$ satisfying $C_jP_jC_j^\dagger=Z_j$. The choices $I$, $H$, and $HS^\dagger$ work for $Z$, $X$, and $Y$, respectively. Set $C_j=I$ for $P_j=I$ and define $C=\bigotimes_jC_j$. Let $Z_S$ have $Z$ at the nonidentity positions of $P$ and $I$ at all other positions. Then
\begin{equation}
C\text{-}P=(I_a\otimes C^\dagger)(C\text{-}Z_S)(I_a\otimes C).
\end{equation}
The corresponding controlled diagonal string is
\begin{equation}
C\text{-}Z_S=\prod_{j:P_j\ne I}\mathrm{CZ}(a,j).
\end{equation}
This product acts as the identity when the ancilla is zero and applies $Z$ at each required position when the ancilla is one. Consequently,
\begin{equation}
C\text{-}P=(I_a\otimes C^\dagger)
\prod_{j:P_j\ne I}\mathrm{CZ}(a,j)(I_a\otimes C).
\end{equation}
Each nonidentity factor needs one controlled-$Z$ gate and at most two single-qubit Clifford operations. An individual single-qubit Clifford operation may comprise several elementary $H$ and $S$ gates.
\end{proof}

\section{Simultaneous Pauli diagonalization}
\label{app:symplectic}
The readout circuit for commuting Pauli generators can be constructed by binary elimination. Represent the generators by $M=[X\mid Z]$, with one row per generator and two columns per qubit. For each qubit, the pairs $(0,0)$, $(1,0)$, $(0,1)$, and $(1,1)$ encode $I$, $X$, $Z$, and $Y$, respectively.

Physical gates correspond to column operations, whereas row operations merely replace the generators of the same commuting group. A Hadamard swaps a qubit's $X$ and $Z$ columns, and a phase gate adds its $X$ column to its $Z$ column. A CNOT applies the corresponding paired changes to the control and target columns. Combining these operations reduces the matrix to $X=0$. Recording the column operations gives the readout circuit without assigning a physical cost to changes of generator basis.
\section{The meaning of the parity ceiling}
\label{app:parity-meaning}
The elementary one-sector inequality $r_X\leq N$ counts columns. Separate conservation of the two spin populations gives the stronger restriction: each Jordan-Wigner $X$ mask lies in an even-weight subspace of dimension $N-1$ within each sector. Combining the two projections gives $r_X\leq2(N-1)$, attained by the CH$_4$ commuting family of Sec.~\ref{sec:search}. This argument covers every subset of the operator's Pauli support. It does not distinguish fragments generated by orbital rotations and therefore is not an orbital-exclusion proof. The correlation-rank theorem of Sec.~\ref{sec:separation} supplies that logically separate obstruction.

\end{document}